\documentclass[11pt,twoside]{article}
\usepackage{fancyhead}
\usepackage{psfrag}
\usepackage{epsfig}
\usepackage{graphicx}
\usepackage{amsmath}
\usepackage{amssymb}
\usepackage{amsthm}
\usepackage{subfigure}
\usepackage{hyperref}
\usepackage{cite}
\usepackage[utf8]{inputenc}

\newtheorem{definition}{Definition}

\newtheorem{example}{Example}

\newtheorem{lemma}{Lemma}
\newtheorem{theorem}{Theorem}
\newtheorem{problem}{Problem}

\DeclareMathOperator{\rank}{rank}
\DeclareMathOperator{\rk}{rk}

\DeclareMathOperator{\defi}{def}

\DeclareMathOperator{\im}{im}

\newcommand{\defeq}{\overset{\defi}{=}}
\newcommand{\Fqm}{\mathbb F_{q^m}}
\newcommand{\Fqn}{\mathbb F_{q^n}}

\newcommand{\Fq}{\mathbb F_{q}}

\newcommand{\dhalffrac}{\left\lfloor \tfrac{d-1}{2}\right\rfloor}
\newcommand{\dhalf}{\left\lfloor (d-1)/2\right\rfloor}

\newcommand{\Gab}[2]{\mathcal{G}(#1,#2)}

\renewcommand{\a}{\mathbf a}
\renewcommand{\c}{\mathbf c}

\renewcommand{\r}{\mathbf r}

\newcommand{\x}{\mathbf x}

\newcommand{\X}{\mathbf X}

\newcommand{\quadbinom}[2]{\ensuremath{
\left[
\begin{matrix}
#1\\
#2
\end{matrix}
\right]
}}
\newcommand{\quadbinomsmall}[2]{\ensuremath{
\left[
\begin{smallmatrix}
#1\\
#2
\end{smallmatrix}
\right]
}}

\begin{document}

\vspace*{5mm}

\noindent
\textbf{\LARGE Bounds on List Decoding Gabidulin Codes
}
\thispagestyle{fancyplain} \setlength\partopsep {0pt} \flushbottom
\date{}

\vspace*{5mm}
\noindent
\textsc{Antonia Wachter-Zeh}\footnote{This work was supported by the German Research Council "Deutsche Forschungsgemeinschaft" (DFG) under Grant No. Bo 867/21-2.} \hfill \texttt{antonia.wachter@uni-ulm.de} \\
{\small Institute of Communications Engineering, Ulm University, Ulm, Germany and Institut de Recherche Math\'ematique de Rennes, Universit\'e de Rennes 1, Rennes, France} \\

\medskip

\begin{center}
\parbox{11,8cm}{\footnotesize
\textbf{Abstract.}
An open question about Gabidulin codes is whether polynomial-time list decoding beyond half the minimum distance is possible or not. 
In this contribution, we give a lower and an upper bound on the list size, i.e., the number of codewords in a ball around the received word. 
The lower bound shows that if the radius of this ball is greater than the Johnson radius, this list size can be exponential and 
hence, no polynomial-time list decoding is possible. The upper bound on the list size uses subspace properties.
}
\end{center}

\baselineskip=0.9\normalbaselineskip

\section{Introduction}
Gabidulin codes \cite{Gabidulin_TheoryOfCodes_1985} can be seen as the analogs of Reed--Solomon (RS) codes in rank metric.
There are several efficient decoding algorithms up to half the minimum rank distance.
However, in contrast to RS codes, there is no polynomial-time decoding algorithm beyond half the minimum distance.
For RS codes, it can be shown that the number of codewords in a ball around \emph{any} received word is always polynomial in the length when the radius of the ball is at most the Johnson radius. 
The Guruswami--Sudan algorithm provides an efficient polynomial-time list decoding algorithm of RS codes up to the Johnson radius.

For Gabidulin codes, there is no polynomial-time list decoding algorithm; it is not even known, whether such an algorithm can exist or not.
An exponential lower bound on the number of codewords in a ball of radius $\tau$ around the received word would prohibit polynomial-time list decoding since the 
list size can be exponential, whereas a polynomial upper bound would show that it might be possible. 
Faure \cite{Faure2006Average} and Augot and Loidreau \cite{AugotLoidreau-JohnsonRankMetric} made first investigations of this problem.

In this paper, we provide a lower and an upper bound on the list size. The lower bound shows that the list size can be exponential in the length when the radius is at least the Johnson radius and therefore in this region, no polynomial-time list decoding is possible. 
The upper bound uses the properties of subspaces and gives a good estimate of the number of codewords in such a ball, but is exponential in the length and therefore does not provide an answer to polynomial-time list decodability in the region up to the Johnson bound.

\section{Preliminaries}
\subsection{Definitions and Notations}
Let $q$ be a power of a prime, let $\Fq$ denote the finite field of order $q$ and let $\Fqm$ be the extension field of degree $m$ over $\Fq$.
We denote $x^{[i]}= x^{q^i}$ for any integer $i$, then a \emph{linearized polynomial}, 
introduced by Ore \cite{Ore_OnASpecialClassOfPolynomials_1933}, 
over $\Fqm$ has the form
$f(x) = \sum_{i=0}^{d_f} f_i x^{[i]}$, 
with $f_i \in \mathbb{F}_{q^m}$. 
If the coefficient $f_{d_f}\neq 0$, we call $d_f \defeq \deg_q f(x)$ the \textit{q-degree} of $f(x)$. 
For all $\alpha_1,\alpha_2 \in \mathbb{F}_{q}$ and $\forall \ a,b \in \mathbb{F}_{q^m}$, the following holds: 
$f(\alpha_1 a+\alpha_2 b) = \alpha_1 f(a)+\alpha_2 f(b)$. 
The (usual) addition and the non-commutative composition $f(g(x))$ (also called \emph{symbolic product}) convert the set of linearized polynomials into a non-commutative ring with the identity element $x^{[0]}=x$. In the following, all polynomials are linearized polynomials. 

Given a basis of $\Fqm$ over $\Fq$, there exists a one-to-one mapping for each vector $\mathbf x \in \mathbb{F}_{q^m}^n$ on a matrix $\mathbf X \in \Fq^{m \times n}$. 
Let $\rank(\mathbf x)$ denote the (usual) rank of $\mathbf X$ over $\mathbb{F}_{q}$ and let 
$\mathcal R (\x) = \mathcal R(\X)$ denote the row space of $\X$ over $\Fq$. 
The kernel of a matrix is denoted by $\ker(\x) = \ker(\X)$ and the image by $\im(\x) = \im(\X)$. 
For an $m \times n$ matrix, if $\dim \ker(\x) = t$, then $\dim \im(\x)=\rank(\x) = n-t$.
Throughout this paper, we use the notation as vector or matrix equivalently, whatever is more convenient.
The \textit{minimum rank distance} $d$ of a code $\mathcal C$ is defined by 
\begin{equation*}
d = \min \lbrace \rank(\mathbf c) \; | \; \mathbf c \in \mathcal C, \mathbf c \neq \mathbf 0 \rbrace. 
\end{equation*}
A \emph{Gabidulin code} can be defined by the evaluation of degree-restricted linearized polynomials as follows.
\begin{definition}[Gabidulin Code \cite{Gabidulin_TheoryOfCodes_1985}]
A linear $\Gab{n}{k}$ Gabidulin code of length $n$ and dimension $k$ over $\mathbb{F}_{q^m}$ for $n \leq m$ is the set of all codewords, which are the evaluation of a $q$-degree restricted linearized polynomial $f(x)$:
\begin{equation*}
\Gab{n}{k} \defeq \lbrace \c  = (f(\alpha_0), f(\alpha_1),\dots, f(\alpha_{n-1}) \big| \deg_q f(x) < k)\rbrace,
\end{equation*}
where the fixed elements $\alpha_0, \dots, \alpha_{n-1} \in \Fqm$ are linearly independent over $\Fq$. 
\end{definition}
Gabidulin codes are \textit{Maximum Rank Distance} (MRD) codes, i.e., they fulfill the rank metric Singleton bound with equality and $d=n-k+1$ \cite{Gabidulin_TheoryOfCodes_1985}. 

The number of $s$-dimensional subspaces of an $n$-dimensional vector space over $\Fq$
is the Gaussian binomial, calculated by
\begin{equation*}
\quadbinom{n}{s} \defeq \prod\limits_{i=0}^{s-1} \frac{q^n-q^i}{q^s-q^i},
\end{equation*}
with the upper and lower bounds 
$q^{s(n-s)}\leq \quadbinomsmall{n}{s} \leq 4 q^{s(n-s)}$.

Moreover, $\mathcal B_{\tau}(\a)$ denotes a ball of radius $\tau$ in rank metric around a word $\a\in \Fqm^n$. The volume of $\mathcal B_{\tau}(\a)$ is independent of its center and is simply the number of $m \times n$ matrices of rank less than or equal to $\tau$.

\subsection{Problem Statement}

\begin{problem}[Maximum List Size]
Let the Gabidulin code $\Gab{n}{k}$ over $\Fqm$ with $n \leq m$ and $d=n-k+1$ be given. Let $\tau < d$. 
Find a lower and upper bound on the maximum number of codewords $\ell$ in the ball of 
rank radius $\tau$ around $\r= (r_0 \ r_1 \ \dots \ r_{n-1}) \in \Fqm^n$. Hence, find a bound on
\begin{equation*}
\ell \defeq \max_{\r \in \Fqm^n}\left(\left|\mathcal B_{\tau}(\r) \cap \mathcal G\right|\right).
\end{equation*}
\end{problem}
For an upper bound, we have to show that the bound holds for \emph{any} received word $\r$, 
whereas for a lower bound it is sufficient to show that there exists one $\r$ for which 
this bound on the list size is valid.

Let $\mathcal L = \{ \c_1,\c_2,\dots,\c_{\ell} \}$ with $|\mathcal L|=\ell$ denote the list of all codewords in the ball 
of radius $\tau$ around $\r$, i.e., $\c_i \in \Gab{n}{k}$ and $\rank(\r-\c_i) \leq \tau$, for $i=1,\dots,\ell$.

\section{A Lower Bound on the List Size}
Faure showed with a probabilistic approach in \cite{Faure2006Average} that the maximum list size of Gabidulin codes with $n=m$
is exponential in $n$ for $\tau \geq n- \sqrt{n(n-d)}$. Our bound slightly improves this value and uses a different proving strategy, 
based on evaluating linearized polynomials. This approach is inspired by Justesen and Hoholdt's \cite{Justesen2001Bounds} and Ben-Sasson, Kopparty, and Radhakrishna's \cite{BenSasson2010Subspace} approaches for bounding the list size of RS codes. 
\begin{theorem}[Lower Bound on the List Size]
Let the Gabidulin code $\Gab{n}{k}$ over $\Fqm$ with $n \leq m$ and $d=n-k+1$ be given. Let $\tau < d$. 
Then, there exists a word $\r \in \Fqm^n$ such that
\begin{equation}\label{eq:listsize_polys}
 \ell \geq\left|\mathcal B_{\tau}(\r) \cap \mathcal G\right|\geq \frac{\quadbinomsmall{n}{n-\tau}}{(q^m)^{n-\tau-k}} 
= q^m q^{\tau(m+n) -\tau^2-md},
 \end{equation}
 and for the special case of $n=m$: 
$\ell \geq q^n q^{2n\tau - \tau^2 - nd}$.
\end{theorem}

 \begin{proof}
 Since $\tau < d = n-k+1$, also $ k-1 < n-\tau $ holds. 
 Consider all monic linearized polynomials of $q$-degree exactly $n-\tau$ with a root space of dimension 
 $n-\tau$, where all roots are in $\Fqn$. There are exactly (see e.g. \cite[Theorem 11.52]{Berlekamp1984Algebraic})
 $\quadbinomsmall{n}{n-\tau}$
such polynomials. Now, let us consider a subset of these polynomials, denoted by $\mathcal P$: all polynomials where
the $q$-monomials of
$q$-degree greater than or equal to $k$ have the same coefficients. 
Due to Dirichlet's principle there exist coefficients such that the number of 
 such polynomials is 
 \begin{equation*}
|\mathcal P| \geq\frac{\quadbinomsmall{n}{n-\tau}}{(q^m)^{n-\tau-k}},
 \end{equation*}
since there are $(q^m)^{n-\tau-k}$ possibilities to choose the 
  highest $n-\tau -(k-1)$ coefficients of a \emph{monic} linearized polynomial over $\Fqm$. 
 
Note that 
the difference between any two polynomials in $\mathcal P$ is a linearized polynomial of $q$-degree
strictly less than $k$ and therefore the evaluation polynomial of a codeword of $\Gab{n}{k}$. 
Let $\r$ be the evaluation of $f(x) \in \mathcal P$ at a basis $\mathcal A = \{\alpha_0,\alpha_1,\dots,\alpha_{n-1}\}$ of $\Fqn$ over $\Fq$:
 \begin{equation*}
 \r = (r_0 \ r_1 \ \dots \ r_{n-1}) = (f(\alpha_0) \ f(\alpha_1) \ \dots \ f(\alpha_{n-1})).
 \end{equation*}
Further, let also $g(x) \in \mathcal P$, then 
 $f(x)-g(x)$ has $q$-degree less than $k$. Let
 $\c$ denote the evaluation of $f(x)-g(x)$ at $\mathcal A$.
Then, $\r-\c$ is the evaluation of $f(x)-f(x)+g(x) = g(x) \in \mathcal P$, whose 
root space has dimension $n-\tau$ and all roots are in $\Fqn$.  
 Thus, $\dim \ker(\r-\c) = n-\tau$ and $\dim \im (\r-\c) = \rk(\r-\c) = \tau$. 
 Therefore, for \emph{any} $g(x) \in \mathcal P$, the evaluation of $f(x)-g(x)$
 is a codeword from $\Gab{n}{k}$ and has rank distance $\tau$ from $\r$. 
 This provides the following lower bound on the maximum list size:
 \begin{equation*}
 \ell \geq \frac{\quadbinomsmall{n}{n-\tau}}{(q^m)^{n-\tau-k}} 
 \geq \frac{q^{(n-\tau)\tau}}{(q^m)^{n-\tau-k}} 
= q^m q^{\tau(m+n) -\tau^2-md},
 \end{equation*}
 and for $n=m$ the special case follows.
 \end{proof}
 This lower bound is valid for any $\tau < d$, but we want to know, which is the smallest value for $\tau$ such that this expression is \emph{exponential} in $n$.
 For $n=m$, we can rewrite \eqref{eq:listsize_polys} by
 \begin{equation*}
 \ell \geq q^{n(1-\epsilon)} \cdot q^{2n\tau - \tau^2 - nd+n\epsilon},
 \end{equation*}
 where the first part is exponential in $n$ for any $0 \leq \epsilon < 1$. 
 The second exponent is positive for 
  \begin{equation}\label{eq:tau_lessjohnson}
 \tau \geq n- \sqrt{n(n-d+\epsilon)} \defeq \tau_{LB}.
  \end{equation}
Therefore, our lower bound \eqref{eq:listsize_polys} 
shows that the maximum list size is exponential in $n$ for $\tau \geq \tau_{LB}$. 
For $\epsilon = 0$, the value $\tau_{LB}$ gives exactly the Johnson radius 
for Hamming metric. 

This reveals a difference between the known limits to list decoding of Gabidulin and RS codes. 
For RS codes, polynomial-time list decoding up to the Johnson radius is guaranteed 
by the Guruswami--Sudan algorithm. However, it is not proven that the Johnson radius is tight for RS codes, i.e., it is not known if the list size is polynomial between the Johnson radius and the known exponential lower bounds (see e.g. \cite{Justesen2001Bounds,BenSasson2010Subspace}).
For Gabidulin codes, we have shown that the maximum list size is exponential for $\tau \geq\tau_{LB}$, which is asymptotically equal to the 
Hamming metric Johnson radius. 

\begin{example}
For the Gabidulin code $\mathcal G (n=12,k=6)$ with $d=7$, the Bounded Minimum Distance decoding radius is $\tau_{BMD} = \dhalf= 3$, 
the lower bound by Faure (equivalent to the Hamming metric Johnson radius) is $\tau_J = \lceil 4.2 \rceil = 5$ and 
\eqref{eq:tau_lessjohnson} with $\epsilon = 0.9$ gives $\tau_{LB} = \lceil 3.58 \rceil  = 4$. This means for this 
code of rate $k/n=1/2$, no polynomial time list-decoding beyond $\tau_{BMD}$ is possible.
\end{example}

\section{An Upper Bound on the List Size}

  The following lemma shows that the row spaces of $\r-\c_i$ and $\r-\c_j$, $\c_i,\c_j \in \mathcal L$, $i \neq j$, have no $(2\tau-d+1)$-dimensional subspace in common.
  \begin{lemma}\label{lem:subspace2taud1}
  Let $\tau < d$ and let $\r \in \Fqm^n$. Let $\c_i$, for $i = 1,\dots,\ell$, be codewords of 
  the Gabidulin code $\Gab{n}{k}$ with minimum rank distance $d$ and let $\rk(\r-\c_i) \leq \tau$ hold for all $i=1,\dots,\ell$. Let $\rk(\r-\c_i) =t_i\leq \tau$ and $\rk(\r-\c_j) =t_j\leq \tau$, $i \neq j$.
   Then, the row spaces of $(\r-\c_i)$ and $(\r-\c_j)$ have no subspace of dimension at least $t_i +t_j -d+1$ in common, for $\dhalf < t_i,t_j \leq \tau$.
  \end{lemma}
  \begin{proof}

  Assume, there exist $(\r-\c_i)$ and $(\r-\c_j)$, with $\rank(\r-\c_i)= t_i$, $\rank(\r-\c_j)=t_j$, $i \neq j$, such that their row spaces contain the same 
  subspace of dimension at least $(t_i+t_j -d+1)$. Then, 
   \begin{align*} 
   \dim(\mathcal R(\c_i-\c_j)) &=
   \dim(\mathcal R(\r-\c_i-\r+\c_j)) 
 \leq   \dim\left(\mathcal R
   \left(
   \begin{matrix}
   \r-\c_i\\
   \r-\c_j
   \end{matrix}\right)\right)\\
   & \leq t_i+t_j -(t_i+t_j-d+1) = d-1,
  \end{align*}
  which is a contradiction to $\rk(\c_i-\c_j) \geq d$.
  \end{proof}
  This means in particular, if $\rk(\r-\c_i) = \rk(\r-\c_j) = t \leq \tau$, they have no subspace of dimension at least
  $2t-d+1$ in common. 
 Based on this lemma, we obtain the following upper bound on the maximum list size.
 
 \begin{theorem}[Upper Bound on the List Size]
 Let the Gabidulin code $\Gab{n}{k}$ over $\Fqm$ with $n \leq m$ and $d=n-k+1$ be given. Let $\tau < d$. 
 Then, for any word $\r \in \Fqm^n$ and hence, for the maximum list size, the following holds
 \begin{align}
  \ell &=\max_{\r \in \Fqm^n}\left(\left|\mathcal B_{\tau}(\r) \cap \mathcal G\right|\right)
  \leq \sum\limits_{t=\dhalf+1}^{\tau} \frac{\quadbinomsmall{n}{2 t + 1-d}}{\quadbinomsmall{t}{2 t + 1-d}}\label{eq:listsize_polys_upper}\\
  &\leq  4 \sum\limits_{t=\dhalf+1}^{\tau} q^{(2t-d+1)(n-t)}
  \leq 4\left(\tau-\dhalffrac\right)\cdot q^{(2\tau-d+1)(n-\dhalf-1)}.\label{eq:listsize_polys_upper2}
  \end{align}
  \end{theorem}
 \begin{proof}
 We consider all words in $\Fqm^n$ with $n\leq m$, therefore these words can be seen as matrices in an $n$-dimensional space.
 For any $t$, where $\dhalf \leq t \leq d$, 
there are $\quadbinomsmall{n}{2 t-d+1}$ subspaces of dimension 
 $2 t-d+1$. 
 Let $\r$ be any fixed word in $\Fqm^n$ and all codewords in $\mathcal L$ have $\rk(\r-\c_i) \leq \tau$. 
 Each $\r-\c_i$, for $i=1,\dots,\ell$, of rank $t \leq \tau$ contains $\quadbinomsmall{t}{2 t -d+ 1}$ subspaces of dimension $2t-d+1$.
 
 Due to Lemma~\ref{lem:subspace2taud1}, different $\r-\c_i$ 
 have no $(2 t -d+ 1)$-dimensional subspace in common and therefore 
 there are at most $\quadbinomsmall{n}{2 t -d+ 1}/\quadbinomsmall{t}{2 t + 1-d}$ possible codewords in 
 rank distance $t$ to the word $\r$. We sum this up for $t$ from $\dhalf +1$ up to $\tau$ and we obtain \eqref{eq:listsize_polys_upper}.
 
 With the bounds for the Gaussian binomial and since $\dhalf+1\leq t \leq \tau$, the upper bound from \eqref{eq:listsize_polys_upper2} follows.
 \end{proof}
 
Note that for the special case $\tau = d/2$ and even minimum distance $d$, the upper bound from \eqref{eq:listsize_polys_upper} is the bound from \cite[Equation~(4)]{AugotLoidreau-JohnsonRankMetric}, which is
  \begin{align*}
  \ell &\leq (q^n-1) \frac{q^{n-d/2}-q^{n-d}}{q^n-2q^{n-d/2}+q^{n-d}}
   =(q^n-1) \frac{q^{-\tau}-q^{-2\tau}}{1-2q^{-\tau}+q^{-2\tau}}
   =\frac{q^n-1}{q^\tau-1} = 
   \frac{\quadbinomsmall{n}{1}}{\quadbinomsmall{\tau}{1}}.
   \end{align*}

Thus, we have proved an upper bound on the maximum list size of a Gabidulin code. 
Unfortunately, this upper bound is exponential in $n \leq m $ for any $\tau > \dhalf$ and therefore
does not provide any conclusion if polynomial-time list decoding is possible or not in the region up to the Johnson bound.


\begin{thebibliography}{1}
\providecommand{\url}[1]{#1}
\csname url@rmstyle\endcsname
\providecommand{\newblock}{\relax}
\providecommand{\bibinfo}[2]{#2}
\providecommand\BIBentrySTDinterwordspacing{\spaceskip=0pt\relax}
\providecommand\BIBentryALTinterwordstretchfactor{4}
\providecommand\BIBentryALTinterwordspacing{\spaceskip=\fontdimen2\font plus
\BIBentryALTinterwordstretchfactor\fontdimen3\font minus
  \fontdimen4\font\relax}
\providecommand\BIBforeignlanguage[2]{{%
\expandafter\ifx\csname l@#1\endcsname\relax
\typeout{** WARNING: IEEEtran.bst: No hyphenation pattern has been}%
\typeout{** loaded for the language `#1'. Using the pattern for}%
\typeout{** the default language instead.}%
\else
\language=\csname l@#1\endcsname
\fi
#2}}

\bibitem{Gabidulin_TheoryOfCodes_1985}
E.~M. Gabidulin, ``{Theory of Codes with Maximum Rank Distance},'' \emph{Probl.
  Peredachi Inf.}, vol.~21, no.~1, pp. 3--16, 1985.

\bibitem{Faure2006Average}
C.~Faure, ``{Average Number of Gabidulin Codewords within a Sphere},'' in
  \emph{Tenth International Workshop on Algebraic and Combinatorial Coding
  Theory}, Sept. 2006, pp. 86--89.

\bibitem{AugotLoidreau-JohnsonRankMetric}
D.~Augot and P.~Loidreau, ``Johnson-like bounds for the rank metric,''
  \emph{preprint}, 2011.

\bibitem{Ore_OnASpecialClassOfPolynomials_1933}
O.~Ore, ``{On a Special Class of Polynomials},'' \emph{Transactions of the
  American Mathematical Society}, vol.~35, pp. 559--584, 1933.

\bibitem{Justesen2001Bounds}
J.~Justesen and T.~Hoholdt, ``{Bounds on list decoding of MDS codes},''
  \emph{Information Theory, IEEE Transactions on}, vol.~47, no.~4, pp.
  1604--1609, May 2001.

\bibitem{BenSasson2010Subspace}
E.~Ben-Sasson, S.~Kopparty, and J.~Radhakrishnan, ``{Subspace Polynomials and
  Limits to List Decoding of Reed–Solomon Codes},'' \emph{Information Theory,
  IEEE Transactions on}, vol.~56, no.~1, pp. 113--120, Jan. 2010.

\bibitem{Berlekamp1984Algebraic}
E.~R. Berlekamp, \emph{{Algebraic Coding Theory, Revised Edition (M-6) (No.
  M-6)}}, revised~ed.\hskip 1em plus 0.5em minus 0.4em\relax Aegean Park Pr,
  June 1984.

\end{thebibliography}
\end{document}